\documentclass[lettersize,journal]{IEEEtran}
\usepackage[numbers,sort&compress]{natbib}
\usepackage{amsmath,amssymb,amsfonts}
\usepackage[ruled,vlined]{algorithm2e}
\usepackage{graphicx}
\usepackage{textcomp}
\usepackage{xcolor}
\usepackage{colortbl}
\usepackage{fancyhdr}
\usepackage{multirow}
\usepackage{diagbox} 
\usepackage{threeparttable}
\usepackage{booktabs} 
\usepackage{algpseudocode}
\usepackage{amsthm}
\usepackage{comment}
\usepackage{enumitem}
\usepackage{subcaption}
\usepackage{cleveref}
\usepackage{url}
\usepackage{pifont}
\definecolor{blue}{RGB}{10, 10, 200}
\newcommand{\cmark}{\ding{51}}%
\newcommand{\xmark}{\ding{55}}%

\newcommand{\bheading}[1]{{\vspace{4pt}\noindent{\textbf{#1}}}}

\newcounter{note}[section]

\newcommand{\secref}[1]{\mbox{Sec.~\ref{#1}}\xspace}

\newcommand{\figref}[1]{\mbox{Fig.~\ref{#1}}}

\newcommand{\ignore}[1]{}

\newcommand{\etc}{\textit{etc.}\xspace}
\newcommand{\ie}{\textit{i.e.}\xspace}
\newcommand{\eg}{\textit{e.g.}\xspace}

\newcommand{\etal}{\textit{et al.}\xspace}

\newcommand{\sysname}{\textsc{TeeRollup}\xspace}
\newcommand{\strawman}{\textsc{SRollup}\xspace}

\newcounter{packednmbr}

\newenvironment{packeditemize}{
\begin{list}{$\bullet$}{
\setlength{\labelwidth}{0pt}
\setlength{\itemsep}{2pt}
\setlength{\leftmargin}{\labelwidth}
\addtolength{\leftmargin}{\labelsep}
\setlength{\parindent}{0pt}
\setlength{\listparindent}{\parindent}
\setlength{\parsep}{1pt}
\setlength{\topsep}{1pt}}}{\end{list}}

\theoremstyle{definition}
\newtheorem{theorem}{Theorem}
\newtheorem{lemma}{Lemma}

\begin{document}{
\title{\sysname: Efficient Rollup Design Using Heterogeneous TEE}

\author{
    Xiaoqing Wen, 
    Quanbi Feng, 
    Hanzheng Lyu,
    Jianyu Niu,~\IEEEmembership{Member, ~IEEE,} \\
    Yinqian Zhang,~\IEEEmembership{Member, ~IEEE,} 
    Chen Feng,~\IEEEmembership{Member, ~IEEE}
    \IEEEcompsocitemizethanks{

        \IEEEcompsocthanksitem  
        The work of Xiaoqing Wen and Chen Feng is supported in part by the NSERC Discovery Grants RGPIN-2016-05310.
        Jianyu Niu was supported in part by the NSFC under Grant 62302204.
        
        Xiaoqing Wen, Hanzheng Lyu, and Chen Feng are with Blockchain@UBC and the School of Engineering, The University of British Columbia (Okanagan Campus), Kelowna, BC, Canada. Email: \{xqwen, hzlyu\}@student.ubc.ca and chen.feng@ubc.ca. 
        
        Quanbi Feng, Jianyu Niu, and Yinqian Zhang are with the Research Institute of Trustworthy Autonomous Systems and the Department of Computer Science and Engineering, Southern University of Science and Technology, Shenzhen, China.
        Email: 12011501@mail.sustech.edu.cn, niujy@sustech.edu.cn and yinqianz@acm.org. Corresponding author: Jianyu Niu. 
    }
}
\sloppy

\markboth{Journal of \LaTeX\ Class Files,~Vol.~14, No.~8, August~2021}%
{Shell \MakeLowercase{\textit{et al.}}: A Sample Article Using IEEEtran.cls for IEEE Journals}

\maketitle

\begin{abstract}
Rollups have emerged as a promising approach to improving blockchains' scalability by offloading transaction execution off-chain. 
Existing rollup solutions either leverage complex zero-knowledge proofs or optimistically assume execution correctness unless challenged.
However, these solutions suffer from high gas costs and significant withdrawal delays, hindering their adoption in decentralized applications.
This paper introduces \sysname, an efficient rollup protocol that leverages Trusted Execution Environments (TEEs) to achieve both low gas costs and short withdrawal delays.
Sequencers (\ie, system participants) execute transactions within TEEs and upload signed execution results to the blockchain with confidential keys of TEEs. 
Unlike most TEE-assisted blockchain designs, \sysname adopts a practical threat model where the integrity and availability of TEEs may be compromised.
To address these issues, we first introduce a distributed system of sequencers with heterogeneous TEEs, ensuring system security even if a certain proportion of TEEs are compromised. 
Second, we propose a challenge mechanism to solve the redeemability issue caused by TEE unavailability.
Furthermore, \sysname incorporates Data Availability Providers (DAPs) to reduce on-chain storage overhead and uses a laziness penalty mechanism to regulate DAP behavior. 
We implement a prototype of \sysname in Golang, using the Ethereum test network, Sepolia.
Our experimental results indicate that \sysname outperforms zero-knowledge rollups (ZK-rollups), reducing on-chain verification costs by approximately 86\% and withdrawal delays to a few minutes.
\end{abstract}


\section{Introduction} \label{sec:intro} 
\IEEEPARstart{S}ince the first adoption in Bitcoin~\cite{nakamoto2008bitcoin} in 2008, blockchain technology has experienced remarkable growth in extensive decentralized applications~\cite{10316333, 9732238}.
However, increasing application requests make blockchains encounter severe congestion due to their low throughput, \ie, scalability issues~\cite{dinh2018untangling}.
To improve the scalability, rollups~\cite{starkex, optimism} that offload a substantial portion of transaction execution off-chain emerged as a promising approach. 
Specifically, they can aggregate the execution result of a batch of off-chain transactions into a single transaction (\ie, a \textit{state transition} of the rollup) on the blockchain (commonly referred to as the \textit{main chain}), whereas the main chain only needs to verify the validity of the uploaded state transitions. 
Due to the promising potential, rollups have garnered significant attention and research efforts~\cite{kotzer2024sok}. 
As of January 2025, statistics from L2BEAT~\cite{l2beat} reveal 51 rollup projects in the market, collectively holding a locked value of approximately \$57 billion. 
Among these, Arbitrum~\cite{arbitrum} alone boasts 37 million active accounts.

Current rollup solutions can be classified into two categories based on verification methods:
zero-knowledge rollups (\ie, ZK-rollups)~\cite{starkex, starknet, scroll} and optimistic rollups (\ie, OP-rollups)~\cite{optimism, arbitrum}.
However, these solutions either have high on-chain costs or prolonged withdrawal delays, posing significant challenges to their adoption, particularly in cost-sensitive and time-critical applications.
Specifically, ZK-rollups rely on complex zero-knowledge proofs for execution validation on the main chain, but this approach incurs significant gas costs.
For instance, verifying a proof generated by ZK-SNARK~\cite{zksnark} on Ethereum consumes over 600K gas (approximately \$39).
Additionally, generating these proofs is time-intensive and poses challenges for developing, deploying, and supporting smart contracts.
In contrast, OP-rollups execute transactions optimistically without immediate verification but face delays in fund withdrawals.
This is because OP-rollups introduce a dispute period to ensure security, allowing anyone to detect and report invalid state transitions.
For example, systems like Arbitrum and Optimism impose a withdrawal delay of up to one week to maintain security~\cite{arbitrum, optimism}.}

{To break the trade-off between fast withdrawal and low on-chain cost, we propose \sysname}, which leverages nodes equipped with TEEs to execute transactions off-chain. 
Specifically, \sysname only requires the main chain to verify the TEE-generated signatures. 
Besides, users can utilize remote attestation during initialization to verify the code running within the TEEs, at which point a key pair is generated inside the TEE.
Then, nodes execute transactions within the TEE and generate state transitions signed by the TEE, which are verified by the blockchain.
Ideally (\ie, without attacks), the integrity provided by TEEs ensures correct execution, preventing malicious nodes from modifying user account status and generating fault state transitions. 
Thus, the main chain can trust the state transition of rollup through simple signature verification.

Unfortunately, directly using TEEs introduces security risks due to their vulnerabilities and unavailability.
First, prior research has identified vulnerabilities to compromise TEEs' integrity and confidentiality~\cite{li2021cipherleaks, li2019exploiting, li2021crossline}.
Once a TEE is compromised, a malicious node can disrupt correct transaction execution and generate faulty states with valid signatures from the TEE.
To address this issue, \sysname adopts a distributed system of nodes with heterogeneous TEEs (\eg, Intel SGX~\cite{mckeen2013innovative}, Intel TDX~\cite{tdx}, AMD SEV~\cite{li2022systematic}, ARM TrustZone~\cite{trustzone}, and Hygon CSV~\cite{csv}, \etc).
In \sysname, a state transition must be signed by multiple TEEs to ensure that a single compromised TEE cannot arbitrarily generate a faulty state transition with valid signatures. 
The states of \sysname are organized in a chained structure, where each new state transition is based on the latest state recorded on the main chain. 
Nodes compete for the right to submit state transitions.
However, the main chain will only accept the first valid state transition with sufficient signatures from TEEs, while subsequent submissions are disregarded. 
Additionally, we utilize on-chain registration and attestation, eliminating the need for mutual attestation between heterogeneous TEEs.

Second, TEEs' I/O is manipulated by their hosts (\ie, hypervisors or software systems), allowing malicious nodes to drop messages to and from their TEEs. 
Since the number of nodes in \sysname is small (\eg, tens of entities), there may not be enough honest nodes to ensure system availability. If TEE availability is disrupted, users cannot redeem their deposited funds promptly, leading to the locking of user funds on the blockchain. 
To mitigate this redeemability issue, \sysname introduces a \textit{challenge mechanism} on the blockchain, enabling users to redeem their deposits without relying on sequencers.
Specifically, users can initiate challenges on the main chain and redeem their deposit on the main chain with proof of their balance if no response is received.
Here, to generate the proof of balance, the metadata of the system state is required, not just the digest. 

Third, to reduce the on-chain storage cost, \sysname also incorporates data availability providers (DAPs) to store metadata off-chain. 
Specifically, it only records the digest of metadata on the main chain and offloads the state data and transaction data to the off-chain DAPs.
However, the DAPs can be lazy and discard the metadata that they are supposed to store to minimize their costs.
Thus, we introduce a collateral scheme to punish dishonest behaviors and a \textit{laziness penalty mechanism} to incentivize DAPs' participation. 

We built a prototype of \sysname with Golang,
developing the on-chain smart contracts using Solidity 0.8.0 and deploying them on the Ethereum test network, Sepolia. Our experiments use mainstream off-the-shelf TEE devices, including Intel SGX, Intel TDX, and Hygon CSV.
However, our protocol does not rely on certain platforms and can be extended to more TEEs like ARM CCA once they are mature.  
These experiments can be easily extended to other virtual machine (VM)-based TEEs.

\bheading{Contributions.} We make the following contributions:
\begin{packeditemize}
\item We introduce \sysname, an efficient rollup design that leverages TEEs for off-chain transaction execution. 
\sysname utilizes a group of heterogeneous TEEs, designed to tolerate the compromise of TEEs. Additionally, \sysname incorporates data availability providers to reduce on-chain storage costs. 

\item We address the redeemability issue stemming from TEEs' unavailability. Our challenge mechanism ensures the redeemability of user funds, even during periods of system crash, because of the unavailability issues of TEEs.
Besides, we implement penalization mechanisms to ensure the diligent operation of data availability providers.

\item We conduct experimentation on the prototyped \sysname to evaluate its performance compared to ZK-rollup and OP-rollup. 
Our experimental results demonstrate that the transaction fee in \sysname is significantly lower at \$0.006 compared to StarkNet's \$0.043, showcasing the efficiency of our approach.
Moreover, \sysname maintains a normal withdrawal time of a few minutes, consistent with StarkNet and better than OP-rollup schemes such as Optimism.
\end{packeditemize}

\bheading{Roadmap.} 
We introduce related work in~\secref{sec:background} and the system model and goal in~\secref{sec:model}. 
The preliminaries are provided in~\secref{sec:preliminary}
We present the system design in~\secref{sec:design}, followed by the security analysis in~\secref{sec:analysis}. 
The system performance is evaluated in~\secref{sec:evaluate}. 
Finally, the paper is concluded in~\secref{sec:conclusion}.

\section{Background}\label{sec:background}

\subsection{Revisiting Rollup Solutions}
Rollup solutions alleviate the blockchains' limitations, including low throughput, high fees, and network congestion, by offloading transactions
processing off-chain~\cite{thibault2022blockchain}. 
Rollup networks are composed of sequencers who aggregate a batch of off-chain transactions into a single state transition.
The main chain only needs to verify the validity of the state transition.
According to the verification methods, current rollup schemes can be divided into two categories: ZK-rollups and OP-rollups, as introduced below.

ZK-rollup models, as embodied by implementations like StarkEx~\cite{starkex}, generate zero-knowledge proofs~\cite{zksnark} for the validity of the off-chain state transition.  
The emergence of solutions like StarkNet~\cite{starknet} and Scroll~\cite{scroll}, involves Ethereum Virtual Machine (EVM) compatible zero-knowledge proof. 
This unique capability extends to the verification of contracts (not just transactions) execution correctness, encompassing both input and output validity in the process.
Thomas \etal further introduce an on-demand ZK-rollup creation service~\cite{lavaur2023modular}, which allows several ZK-rollups to co-exist as groups on the same smart contracts, and application examples for Internet of Everything (IoE)~\cite{lavaur2022enabling}.
However, the on-chain verification of zero-knowledge proofs in ZK-rollup systems necessitates a notable gas expenditure. 
For instance, the gas consumption to verify the zero-knowledge proof generated by zk-SNARK on Ethereum is more than 600K gas (about \$39).
For zk-STARK, the gas consumption is even higher, reaching approximately 2.5M gas (about \$162)~\cite{zksnark}.

OP-rollup solutions, implemented by Optimism~\cite{optimism} and Arbitrum~\cite{arbitrum}, pivot towards an optimistic approach for transaction execution. 
These solutions optimistically assume that transactions are correctly executed by sequencers unless a challenger disputes the execution results.
The dispute can be solved on the main chain.
The advantage here lies in omitting on-chain verification, reducing the gas cost. 
Nevertheless, OP-rollups significantly extend the withdrawal time, as they require sufficient time for challengers to verify the proof.
For instance, Arbitrum and Optimism enforce a withdrawal period of one week for security~\cite{arbitrum, optimism}.

\bheading{Summary.} Existing rollup solutions suffer from high gas costs, low compatibility, and long withdrawal delays.
A comparison between existing rollup solutions and \sysname (proposed in this paper) is provided in Table~\ref{table:comparison}.

\begin{table}[t]
\centering
    \begin{threeparttable}
    \caption{\textbf{Rollup Solutions.}
    The Dec, EVM Comp, and Data Avail are short for Decentralized, EVM Compatibility, and Data Availability, respectively.
    Decentralized property means that the transaction cannot be confirmed by a single sequencer.
    The EVM compatibility denotes the ability to support EVM. 
    For Data Availability, Main Chain and Delegated means the data availability is provided by the main chain and a third party, respectively.
    Withdrawal time signifies the waiting time for users to retrieve the funds they locked in the main chain.}
    \label{table:comparison}
    \begin{tabular}{@{}lcccc@{}}
    \toprule
    \textbf{Rollup Solution}  & \textbf{Dec}  & \textbf{EVM Comp}& \textbf{Data Avail}& \textbf{Withdrawal time} \\
     \midrule
    StarkEx~\cite{starkex}            & {\xmark}  & \xmark & {Main Chain}    & Few minutes\\
    StarkNet~\cite{starknet}            &  \xmark & \cmark &  {Main Chain}  & Few minutes\\
    Scroll~\cite{scroll}            &  \xmark & \cmark &  {Main Chain}  & Few minutes\\
    Optimism~\cite{optimism}            &  {\xmark} & \cmark &  {Main Chain}  & One week\\
    Arbitrum~\cite{arbitrum}            &  {\xmark} & \cmark &  {Main Chain}  & One week\\
    
\midrule
\textbf{\sysname}                      & \cmark  & \cmark & Delegated      & Few minutes \\
\bottomrule
\end{tabular}
\end{threeparttable}
\end{table} 

\subsection{Rollup Data Availability}
One critical aspect of rollup solutions is ensuring data availability, which refers to the accessibility and integrity of transaction data necessary for maintaining the security and trustworthiness of the system~\cite{huang2024data}.
Storing complete data on the main chain inherits the security and data integrity of the main chain, but faces challenges such as network congestion and high gas costs.
For instance, storing a 256-bit integer on the Ethereum smart contract \textup{STORAGE} field costs 8K gas (0.52 USD)~\cite{solidity}.
Therefore, it is impractical to simply permanently store the complete data of the rollup on the main chain.
StarkEx~\cite{starkex} and StarkNet~\cite{starknet} address data availability issues by including aggregated compressed transactions in the \textup{CALLDATA} field of the main chain, which reduces the gas cost~\cite{calldata}.
However, this method does not allow for permanent storage of transactions and consumes more gas compared to only recording the digest of the state on the main chain.
Thus, we choose to only record the digest of rollup state (\ie, the state) on the main chain with complete data (\ie, metadata) stored off-chain by the data availability providers (DAPs).

\subsection{Trusted Execution Environment}
Trusted Execution Environments (TEEs) aim to run applications in a secure environment without leaking secrets to an adversary who controls the computing infrastructure. 
Specifically, TEEs provide \textit{enclave} to run code and the attestation mechanism to prove the correctness of computation.
Influential TEE implementations include Intel SGX~\cite{mckeen2013innovative}, Intel TDX~\cite{tdx}, AMD SEV~\cite{li2022systematic}, ARM TrustZone~\cite{trustzone}, and Hygon CSV~\cite{csv}. 

Recently, TEEs have been widely used in blockchain designs to enhance security, privacy, and performance~\cite{lind2019teechain, tesseract, yin2022bool, pose, xu2022l2chain,9705115, wen2024mecury,engraft, narrator,narrator-pro}. 
Teechain~\cite{lind2019teechain} establishes a payment platform under TEE protection, while Tesseract~\cite{tesseract} and Bool Network~\cite{yin2022bool} employ TEE to ensure the security of cross-chain transactions.
Tommaso \etal~\cite{pose} and Liu \etal~\cite{9705115} leverage the TEE for the privacy of off-chain execution and data collection, respectively.
Other studies~\cite{engraft, narrator, narrator-pro, niu2025achilles,xie2025fides,wang2024tiks} focus on utilizing TEEs to improve the performance of Byzantine fault-tolerant consensus protocols. 
{These works assume that TEEs provide integrity and confidentiality guarantees, where the adversary cannot know the private key inside the TEE.
}

However, many studies indicate that TEEs are vulnerable to various attacks, such as side-channel attacks ~\cite{li2021cipherleaks}
, unprotected I/O~\cite{li2019exploiting}
, and ASID abuses~\cite{li2021crossline}.
Thus, in our work, we assume TEEs can be compromised, \ie, no integrity and confidentiality properties, which differs from prior work with the perfect assumption of TEEs~\cite{lind2019teechain, tesseract, yin2022bool, pose, xu2022l2chain,9705115, engraft, narrator,narrator-pro}.  
{
In particular, the adversary can steal the private keys for signing messages from the compromised TEE.
To tolerate malicious behaviors of compromised TEEs, one approach is to use a distributed system of heterogeneous TEEs.
}
This heterogeneity can arise from using TEEs from different vendors, different types of processors, or different levels of security requirements~\cite{pouyanrad2023end}.
Since different TEE platforms adopt different hardware and software architectures, it is difficult to breach their security simultaneously. 
Besides, when a TEE is compromised, manufacturers will timely resolve these issues, making security compromise of all TEEs nearly impossible~\cite{liu2015oblivm, lang2022mole}.

{
\subsection{Preliminaries}\label{sec:preliminary}
We introduce the cryptography primitives, Merkle tree, and the model of TEEs used in the \sysname protocol.

\vspace{2mm} \noindent \textbf{Cryptography primitives.}
Our protocol utilizes the public key encryption scheme ($\textsc{GenPK}, \textsc{Enc}, \textsc{Dec}$), a signature scheme ($\textsc{Sign}, \textsc{Verify}$), and a secure hash function $\textsc{H}(\cdot)$.
All messages are signed by their respective senders, with a message $m$ signed by party $p$ written as $\sigma \leftarrow \textsc{Sign}(m;p)$, where $\sigma$ is the signature.

\vspace{2mm} \noindent \textbf{Merkle tree.}
Merkle tree is a binary tree data structure in which each leaf node contains the hash of a data block (\eg, the key-value pair), and each non-leaf node contains the hash of its child nodes. 
Constructed bottom-up, a single hash at the root can be generated, representing the entire dataset's integrity. 
The root of a Merkle Tree is integrated into each block, serving as a comprehensive state digest, enabling sequencers to maintain and verify the state digest efficiently.
The Merkle Tree allows leaves to generate a Merkle path for verifying specific key-value pairs under a root, denoted as $\delta \leftarrow root.proof(\left<k,v\right >)$.
Anyone can verify this pair of key-value through $verifyproof(\delta, \left<k,v\right >)$.

\vspace{2mm} \noindent \textbf{Trusted Execution Environments (TEEs).}
The sequencers can instruct their TEEs to create new enclaves, \ie, new execution environments running a specific program.
We follow prior work~\cite{pass2017formal} to model the functionality of TEEs. 

\begin{packeditemize}
    \item $eid \leftarrow install(prog)$: Installing the program $prog$ as a new enclave with a unique enclave ID, $eid$ within the TEE.

    \item $(outp, \sigma) \leftarrow resume(eid, inp)$: 
    Resuming the enclave from the crash with ID $eid$ to execute program $prog$ using the input $inp$. $\sigma$ is the TEE's endorsement, confirming that $outp$ is the valid output. 
\end{packeditemize}

For uncompromised TEE, the verification of the output $outp$ is trusted by users when $verify(\sigma, eid, prog, outp) = 1$, indicating the successful attestation of the TEE.
Besides, we assume an attestation API for TEE to generate an attestation quote $\rho \leftarrow attest(eid, prog)$ proving that the program $prog$ has been installed in the enclave $eid$.
And $\rho$ can be verified through the method $verifyquote(\rho)$. 
}

\section{Problem Statement and System Model} \label{sec:model}
We first present the system model and threat model. Then, we introduce the associated challenges and solutions. 
Table~\ref{table:notation} lists the frequently used notations.

\begin{table}[t]
    \footnotesize
    \centering
    \setlength{\abovecaptionskip}{0cm}
    \caption{\textbf{Summary of notations.}}
    \label{table:notation}
    \scalebox{0.95}{
    \begin{tabular}{@{}m{0.9cm}l|m{0.9cm}l@{}}
        \toprule[1pt]
        Term    & Description  &  Term    & Description    \\
        \midrule
        $\mathcal{M}$       &Main chain                    & $st_h$           & State of \sysname\\
        $h$               & Height of $st_h$  & TSC                 & \sysname smart contract \\
        MSC                 & Manager smart contract       & $n$               & Number of sequencers\\
        TToken              & Tokens of \sysname           & $p_i$             & Sequencers in \sysname \\
        $T_{addr}$          & TSC account on $\mathcal{M}$ & $\eta_i$          & Enclave for sequencer $p_i$\\
         $(pk_i, sk_i)$    & Key pair for enclave $i$ & $\tau$            & Timer for resolve challenge  \\
        \bottomrule[1pt]
    \end{tabular}}
    \vspace{-0.4cm}
\end{table}

\subsection{System Model}
We consider a system of $n$ sequencers, denoted by the set $\mathcal{G} = \{p_1, p_2, ..., p_n\}$.  
We follow the existing assumption~\cite{eigenlayer}, where each sequencer $p_i$ is equipped with one type of TEE platform (\eg, Intel SGX, Intel TDX, etc.), denoted by $\eta_i$. 
{We assume all sequencers can be malicious by controlling the host of TEEs with at most $f$ sequencers' TEEs are compromised at any time.}
Here, since TEE machines adopt different architectures and are produced by various manufacturers, it is difficult to breach the security of $f$ ones at the same time. 
The parameter $f$ can be determined by the rollup service provider. There is a public/private key pair of the sequencer $p_i$, denoted by $(pk_i, sk_i)$, in which the private key is only stored and used inside  TEE.

\subsection{Threat Model}\label{subsec:threatmodel}
We model the malicious sequencers as Byzantine adversaries, \ie, they can behave in arbitrary ways.
{To capture the different types of threats each component may introduce, we define the malicious behavior of the sequencer and its TEE separately. 
The threat model of sequencers consists of four cases: the honest sequencers with either compromised or uncompromised TEEs, and the malicious sequencers with either compromised or uncompromised TEEs.}
\begin{packeditemize}
    \item \textbf{Case 1: Honest sequencers with uncompromised TEEs.} The honest sequencers with their TEEs do not deviate from the protocol.

    \item \textbf{Case 2: Honest sequencers with compromised TEEs.} The honest sequencers with their TEEs do not deviate from the protocol.
    
    \item \textbf{Case 3: Malicious sequencers with uncompromised TEEs.} The malicious sequencers have full control over the operating system of their TEEs, including root access and control over the network.
    The sequencers can arbitrarily launch, suspend, resume, terminate, and crash TEEs at any time. 
    Besides, the sequencers can delay, replay, drop, and inspect the messages sent to and from TEEs, \ie, manipulating input/output messages of TEEs.
    In other words, the TEE cannot guarantee availability due to these manipulations.
    
    \item \textbf{Case 4: Malicious sequencers with compromised TEEs.} The sequencers have full control of the enclave and know the private key $sk_i$ inside the compromised TEE. In this case, TEEs have no integrity or confidentiality guarantees.
    We assume static corruption by the adversary, where a ﬁxed fraction of all sequencers' TEEs is compromised.
\end{packeditemize}

Existing TEE-based blockchain research~\cite{pose, tesseract} only focuses on cases 1 and  3, without considering compromised TEEs. 
We extend this by considering up to $f$ compromised TEEs of sequencers and scenarios where all sequencers can be malicious.
For data availability providers,  we follow the existing assumption in~\cite{tas2023cryptoeconomic} that they are all rational.

\bheading{Main chain model.} \label{sec:chaninmodel} 
We assume the \sysname is built on the main chain $\mathcal{M}$.
The main chain provides finality (\ie, once a transaction is included in a block, it is considered final)~\cite{zhang2020analysis} and enables smart contracts~\cite{solidity}. 
The finality of a transaction can be verified by checking whether the including block satisfies certain public rules~\cite{ethereum} (\ie, the block is correctly formatted, contains valid signatures, and so on).
Smart contracts can access the current time using the method $block.timestamp$ and provide cryptographic schemes including hash computing and ECDSA encryption~\cite{solidity}.

\subsection{System Goals}\label{subsec:goal}
The rollup system \textit{issues} tokens (denoted by TTokens), and allows clients to \textit{transfer} tokens in the rollup or \textit{redeem} tokens on the main chain.
A client on the main chain $\mathcal{M}$ with account $c_{addr}$ can enter rollup by depositing on the account $T_{addr}$ held by the rollup.
Then, the rollup issues the TTokens and locks the deposit on the account $T_{addr}$.
{Clients can submit transactions to the sequencers to transfer within the account of the rollup.}
Sequencers handle transactions and update the state of the rollup to the blockchain $\mathcal{M}$.

{The state of the rollup can be denoted as $st_h$, where $h$ is the sequence number (referred to as the height) of a state $st$. }
The processing of transactions inside rollup can be defined as $st_{h+1} \leftarrow execute(st_h, txs_{h+1})$, where $txs_{h+1}$ is a batch of transactions.
The rollup system has to satisfy the following two properties:

\begin{packeditemize}
\item \textit{Correctness:} \sysname must ensure the integrity and correctness of state transitions on the main chain.

\item \textit{Redeemability:}
Any client can redeem their TTokens in \sysname.
\end{packeditemize}

\subsection{Design Challenges} \label{subsec:strawman}
To better outline the \sysname design, we first introduce a strawman solution (referred to as \strawman) that leverages TEE to shield the sequencer from malicious behavior.
In \strawman, a single sequencer collects and executes the transactions from clients and submits the execution result to the main chain.
While \strawman always provides the fundamental functions of issuing, transferring, and redeeming TTokens, it does not achieve the \textit{Redeemability} defined in~\secref{subsec:goal}.
This is because the \strawman does not consider the compromised TEE and potential maliciousness of sequencers and data availability providers.

\bheading{Compromised TEE.}
As mentioned in ~\secref{subsec:threatmodel}, TEEs can be compromised, allowing malicious sequencers to acquire the secret keys inside the compromised TEEs. 
Malicious sequencers with compromised TEEs can forge false transactions or execute transactions incorrectly to generate incorrect states with valid signatures. 
Consequently, a single TEE cannot be fully trusted by the system.
To mitigate TEE compromises, \sysname employs a group of sequencers equipped with independent and heterogeneous TEEs. 
Sequencers validate and sign the state transitions within their TEEs, competing for the right to submit them. 
{The main chain will only accept the first received valid state transition, based on the latest state, along with $f+1$ signatures. 
Additionally, \sysname leverages on-chain attestation during registration, which eliminates the need for pairwise attestation between heterogeneous TEEs, thereby reducing message complexity.
}

\bheading{Malicious sequencers.} 
According to our threat model (see~\secref{subsec:threatmodel}), the operating system running TEEs is fully controlled by their sequencers. 
Malicious sequencers can crash their TEEs, censor transactions, or selectively filter out certain client requests. 
In this case, \sysname cannot process client requests, preventing clients from redeeming their deposits and compromising \textit{redeemability}.
To address this, \sysname introduces a challenge mechanism to ensure redeemability and enable on-chain settlement when TEEs fail to provide service. 
If a client finds their transaction unprocessed for an extended period, they can submit the transaction data to the main chain via a challenge method implemented in the smart contract. 
When a challenge is initiated, a timer starts. 
If the sequencer fails to respond within the allotted time, the system is deemed unavailable, and all the deposits will be settled on the main chain.
This mechanism incentivizes sequencers to respond promptly to challenges, preventing system-wide settlement while ensuring clients can redeem their funds even during system unavailability.

\bheading{Lazy data availability providers.} 
The states of \sysname are recorded on the main chain. 
Therefore, the correctness of the metadata provided by DAPs can be verified by the $state$ recorded on the smart contract (\ie, DAPs cannot forge account balance, transaction data, or other metadata).
However, the DAPs can show laziness, discarding the data they are supposed to store.
Thus, clients and sequencers cannot get metadata of the $state$ on smart contracts.
\sysname introduces collateral and laziness penalty mechanisms to incentivize DAPs to behave diligently. 
{
DAPs need to provide collateral when entering the network, and any dishonest behavior will result in slashing penalties.
For the slashing mechanism, we adopt the design proposed in~\cite{tas2023cryptoeconomic}, while incorporating the parameters of \sysname.
}

\begin{figure}[t]
    \centering
    \includegraphics[width=9cm]{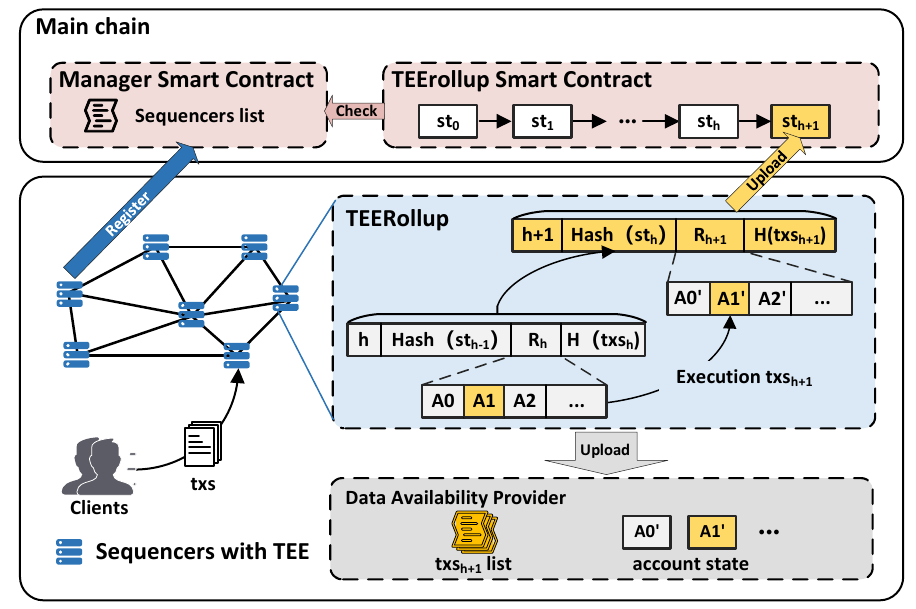}
    \caption{An overview of \sysname architecture. The red components represent smart contracts on the main chain, the blue components indicate operations executed within the TEE, and the gray components correspond to the storage provided by data availability providers.}
    \label{fig:architecture}
\end{figure}

\section{\sysname Design}\label{sec:design}
\subsection{Overview}
\figref{fig:architecture} shows the architecture of \sysname, a rollup based on TEE.
\sysname considers four roles: the \sysname Smart Contract (TSC) responsible for record-keeping on the main chain, the Manager Smart Contract (MSC) for managing sequencers, the TEE-equipped sequencer committee processing off-chain transactions, and the data availability provider (DAP) storing metadata for \sysname. 
\begin{packeditemize}
    \item \textbf{\sysname smart contract (TSC).} TSC records the \sysname state proof in the form of a state (rather than complete metadata) to reduce on-chain storage cost.

    \item \textbf{Manager smart contract (MSC).} MSC manages registered TEEs of sequencers, especially assigning an ID for them and recording the private key $m_{pk}$ inside the TEE.

    \item \textbf{Sequencer.} Sequencer primarily collects and executes transactions from clients, submits the state to TSC, and sends metadata to DAP.

    \item \textbf{Data availability provider (DAP).} DAP accepts and stores the metadata sent by sequencers to ensure access to system metadata at any time. The metadata includes the state of the system (\ie, account balances), the update history, and transaction data.
\end{packeditemize}

There are four key procedures in \sysname. 
First, during the initialization phase, the service provider deploys both the TSC and the MSC on the blockchain $\mathcal{M}$.
Subsequently, sequencers register their TEEs on the MSC (see~\secref{subsec:register}). 
Second, clients deposit funds on the TSC and request the issuance of TTokens in \sysname, which can be used for transfers. 
{Sequencers receive transactions from clients, process these transactions, and generate an updated state. 
Then, the sequencers broadcast the new state along with metadata and collect votes. 
Once enough votes are collected, the sequencers submit the state of \sysname to the TSC and send the metadata to the DAPs (see~\secref{subsec:normal}).}
Third, if clients' transactions are not processed in time, clients can challenge the TSC to enforce the transaction. 
Even if the entire system cannot respond, they can get a refund directly from the TSC to avoid the property loss  (see~\secref{subsec:challenge}).
Fourth, to better regulate the behavior of DAPs and prevent laziness, we use the incentive method of collateral plus laziness punishment (see~\secref{subsec:dap}). 

{
\subsection{{Data Structure}}\label{subsec:data}

\bheading{State format.}
In \sysname, the state acts as a digest of the metadata, and anyone can verify the transactions, balance, and other information of \sysname using this digest. 
The state $st_{h}$ on height $h$ 
has the following structure:
$$st_{h} := \left<h, H(st_{h-1}), R_{h}, H(txs_h)\right>$$
where $H(st_{h-1})$ refers to the hash of the previous state $st_{h-1}$, $R_{h}$ is the Merkle root of the account tree of \sysname, and $H(txs_h)$ is the hash of the transactions list executed to generate the $st_{h}$.
Account tree is a Merkle Tree that stores the clients' accounts in \sysname, denoted as $A_h$.
Account tree $A_h$ consists of the account address and the balance for all accounts.
{ Specifically,  a standard Merkle-Tree is used for the account tree structure. 
We acknowledge that a Merkle Patricia Tree (MPT) could offer additional optimizations, especially for sparse datasets. 
Future work could explore the potential benefits of integrating MPT to improve performance and scalability.}

In \sysname, the states are formatted as a chain structure by following Bitcoin~\cite{nakamoto2008bitcoin},  Ethereum~\cite{ethereum}, and some BFT protocols~\cite{buchman2016tendermint}.
Every state contains the hash of the previous state, and the state can be indexed by height (\ie, the distance from the initial state).
{If there is a path from $st_{h+1}$ to $st_{h}$ then we say that $st_{h+1}$ extends $st_{h}$  }
In \sysname, there is only one state at each height.
Once the latest state is obtained, the correctness of any historical state can also be verified.

\bheading{Vote and certificate.}
\sysname allows sequencers to compete for the right to submit states. 
Specifically, any sequencer can become the leader, broadcast the state, and collect votes.
However, the TSC will only accept the first state that arrives at one height.
A vote $v_i$ from the sequencer $p_i$ of state $st_h$ has the following structure:
$$v_i := \left< H(st_h), pk_i, \sigma_i \right>$$
where $H(st_h)$ is the hash of $st_h$, $pk_i$ is the public key for sequencer $p_i$, and $\sigma_i$ is a signature created by the sequencer $p_i$ over $H(st_h)$.
Here, the signature is generated inside the TEE of sequencer $p_i$ with the private key $sk_i$. 
If there is a set of $f+1$ signatures, it forms the quorum certificate (QC) for the state.
Here, a QC can be implemented as multi-signatures or aggregated signatures.

\begin{figure*}[t]
    \centering
    \includegraphics[width=18cm]{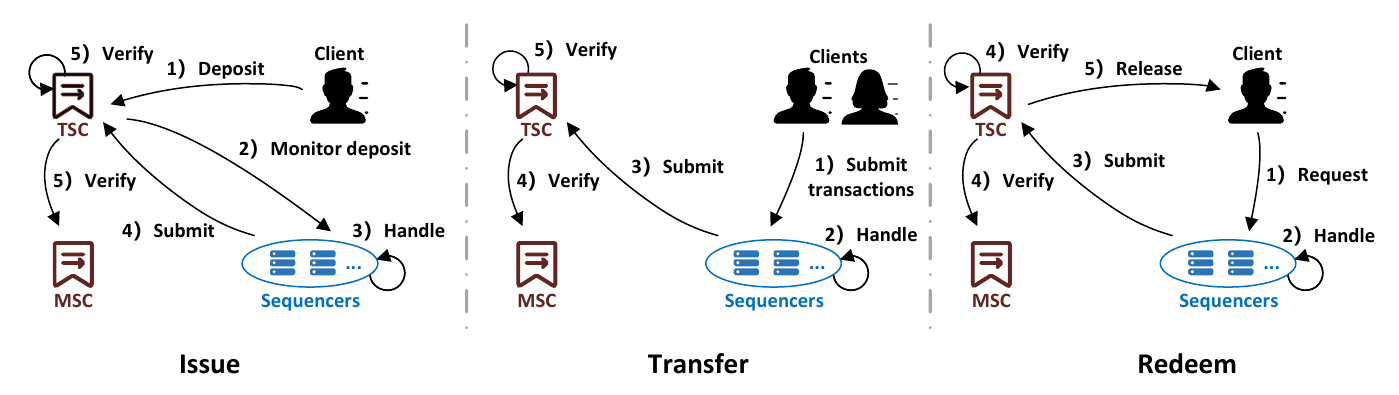}
    \caption{An overview of the main functions provided by \sysname.}
    \label{fig:protocol}
    \vspace{-4mm}
\end{figure*}

\subsection{Sequencers Registration and Configuration}\label{subsec:register}
\sysname implements the MSC to manage the sequencer committee, which provides the registration and configuration for sequencers.

\bheading{Enclave registration.}
Sequencers equipped with TEEs can participate in \sysname by creating an enclave $\eta_i$ and registering it on the MSC.
To clarify the process, we present an example of a sequencer $p_i$ registering on the MSC for clarity.
First, the sequencer $p_i$ creates an enclave $\eta_i$ and installs the \sysname program $prog$ inside it. 
Upon the creation of $\eta_i$, an asymmetric key pair $(pk_i, sk_i)$ is generated. 
Notably, the public key $pk_i$ also serves as an account address for $\eta_i$ on the main chain~\cite{ethereum}.
The secret key $sk_i$ is securely stored within $\eta_i$, and the public key $pk_i$ is returned as output to the sequencer $p_i$. 

Then, the sequencer $p_i$ registers $\eta_i$ by invoking \textsf{Register} with $\left<\eta_i, \rho_i, pk_i \right>$ on the MSC. 
Here, $\rho_i$ represents an attestation proof generated by $\eta_i$, which certifies its correct execution of the program $prog$.
The {MSC verifies} that $verifyquote(\rho_i)$ = 1~\cite{pass2017formal}.
Upon successful validation, the MSC adds the sequencer $p_i$ and $pk_i$ to the sequencer list. 
The registration process guarantees the correct loading of the program $prog$ for all registered enclaves and ensures that the secret key $sk_i$ remains confidential during registration. 
{Since the heterogeneous TEEs can be registered and attested on the MSC, they can rely on the MSC to establish trust without needing mutual attestation. 
As a result, sequencers with heterogeneous TEEs can trust each other and use key pair $\left<pk_i, sk_i \right>$ to generate signatures between TEEs.}
There is no need to re-attest enclaves in subsequent protocol steps.

\bheading{Sequencer configuration.} 
As described earlier, sequencers register their enclaves on the MSC, where their identities and public keys are securely recorded. 
The MSC manages the sequencers and provides authentication services to the TSC. 
Specifically, the MSC maintains a list of sequencers, including their public keys and the committee strategy used by the TSC. 
When the TSC needs to verify messages sent by sequencers, it can invoke the \textsf{VerifyQC} function of the MSC to ensure the authenticity of the QC. 
In this paper, we adopt a voting strategy, where a QC is considered valid if it includes signatures from more than $f+1$ sequencers in the committee.

\subsection{Normal-Case Operations}\label{subsec:normal}

As shown in~\figref{fig:protocol}, \sysname provides issue, transfer, and redeem functions for clients.

\subsubsection{Issue}
Client deposits in the TSC, and the \sysname issues the corresponding amount of TTokens for the client.
The exchange rate between the currency of blockchain $\mathcal{M}$ and TToken is beyond the scope of this paper.
The rate is set to be 1 for simplicity.

\bheading{Step~1.}
The client invokes the \textsf{Deposit} on TSC, denoted as a tuple of $\left<s_{addr}, v\right>_\sigma$, where $s_{addr}$ is the client's address and $v$ is the deposit value (greater than zero), and $\sigma$ is the signature of the client.
{\sysname follows the existing practices of rollup systems by sharing the same key algorithm as the main chain~\cite{zksnark}.
This means that a client holding an account address (\ie, public key) and its corresponding secret key can use the same key pair on both the main chain and the rollup.}

\bheading{Step~2}.
The enclave $\eta_i$ within the sequencer $p_i$ monitors the $deposit$ transaction on TSC.
Enclave $\eta_i$ issues TTokens and changes the balance of account $s_{addr}$ in \sysname.
Enclave $\eta_i$ batches multiple transactions together to optimize efficiency and executes these transactions, generating the new state $st_{h+1}$ and a $lock$ transaction denoted as $\left<id\right>$, where $id$ is the number of the $deposit$.
The sequencer $p_i$ sends the $st_{h+1}$ and $lock$ with the metadata to other sequencers, who vote for the $st_{h+1}$.

\bheading{Step~3}.
After collecting $f+1$ votes, the sequencer $p_i$ generates the QC for $st_{h+1}$ and sends $st_{h+1}$ and the QC to the TSC.
In \sysname, state  $st_{h+1}$ that meets the following conditions can be accepted by the TSC.
First, the height of the last state in TSC is $h$.
Second, $st_{h+1}$ contains the hash value of $st_{h}$, which corresponds to the last state in TSC.
Third, the QC for $st_{h+1}$ is valid.
Once certified, the TSC records $st_{h+1}$ on the main chain. 
Additionally, the TSC locks the client’s $deposit$, which cannot be withdrawn unless the client exits \sysname. 
The update of $st_{h+1}$ and the locking of the $deposit$ are performed in a single operation on the TSC to ensure atomicity.

\bheading{step~4}. 
After $st_{h+1}$ and the lock transaction is confirmed on the main chain, the sequencer $p_i$ returns the execution result to the clients and sends the metadata to the DAPs.

\subsubsection{Transfer}
Clients transfer their TTokens by sending transactions to the sequencer.

\bheading{Step~1.}
A client submits a transaction $tx$ to the sequencers.
$tx$ can be denoted by tuple of $\left<s_{addr}, r_{addr}, v \right>_\sigma$, where $s_{addr}$ is the sender's address, $r_{addr}$ is the receiver's address, $v$ is the transferred value (greater than zero), and $\sigma$ is the signature of the sender.

\bheading{Step~2}.
The execution of $tx$ changes the balance of accounts $s_{addr}$ and $r_{addr}$.
Transaction $tx$ is batched with other transactions into $txs_{h+1}$.
Enclave $\eta_i$ executes $txs_{h+1}$, and generates the new state $st_{h+1}$ signed with its private key $sk_i$. 
Then, the sequencer $p_i$ broadcasts $st_{h+1}$ along with the transactions $txs_{h+1}$ to others and waits for their votes.

\bheading{Step~3}.
After collecting $f+1$ votes, the sequencer $p_i$ generates the QC for $st_{h+1}$ and sends $st_{h+1}$ and the QC to the TSC.
TSC verifies that QC is valid, and then records $st_{h+1}$.

\bheading{Step~4}. 
If $st_{h+1}$ is confirmed on the main chain, the sequencer $p_i$ returns the execution result to the clients and sends the metadata to the DAPs.

\subsubsection{Redeem}
Client burns TTokens in 
\sysname and TSC refunds to the Client.
To burn TTokens, clients transfer them to the designated account $a_0$, which is a burn-only account.
Account $a_0$ can only receive tokens and cannot transfer them, ensuring that tokens sent to $a_0$ are permanently considered burned.
{
\sysname employs an off-chain redemption method to minimize gas costs. To address redeemability challenges, \sysname incorporates a challenge mechanism, detailed in~\secref{subsec:challenge}. 
}

\bheading{Step~1.}
Clients submit requests to the sequencer $p_i$, and the request can be denoted as $\left<s_{addr}, a_0, v \right>_
{\sigma}$, where $s_{addr}$ is the client's address, $v$ is the transferred value (greater than zero), and $\sigma$ is the signature of the client.

\bheading{Step~2}.
Enclave $\eta_i$ batches $tx$ with other transactions, executes them, and generates the new state $st_{h+1}$ with a $refund$ transaction for the TSC, denoted as $\left< s_{addr}, v\right>$.
Enclave $\eta_i$ signs $st_{h+1}$ and $refund$ with its private key $sk_i$. 
Then, the sequencer $p_i$ sends $st_{h+1}$ with $refund$ to others and waits for their votes.

\bheading{Step~3}.
After collecting $f+1$ votes, the sequencer $p_i$ generates the QC for $st_{h+1}$ and $refund$.
Then, the sequencer $p_i$ submits $st_{h+1}$ and $refund$ along with the QC to the TSC.
If the QC is verified as valid, the TSC records $st_{h+1}$ on the main chain and refunds currency with a value of $v$ to $s_{addr}$.

\bheading{Step~4} 
After $st_{h+1}$ and the $refund$ transaction are confirmed on the main chain, the sequencer $p_i$ returns the execution result to the clients and sends the metadata to the DAPs.

\begin{algorithm}[ht]
\caption{Challenge Algorithm of TRollup Smart Contract}\label{alg:smartcontract}
\begin {algorithmic}[1]

    \State{\textbf{Function}~\textbf{\textsc{Init}}}
    \State{~~~~ContractState $\leftarrow$ Active}
    \State{~~~~$Chal \leftarrow \emptyset$}
    \State{~~~~$Dep \leftarrow \emptyset$}

\State
\State{\textbf{Function~\textsc{Deposit} ($M_{addr}, T_{addr},v$)} }
    \State{~~~~\textbf{Require} ContractState $=$ Active}
    \Statex{~~~~~~~~~~~~~~~$\wedge$ $value > 0$}
    \State{~~~~$id \leftarrow $ H($M_{addr}, block.timestamp$)}
    \State{~~~~$Dep_{id}.sender \leftarrow M_{addr}$}
    \State{~~~~$Dep_{id}.value \leftarrow v$}
    \State{~~~~$Dep_{id}.solved \leftarrow false$}
    \State{~~~~Trigger Deposit \textbf{event}}

\State
\State{\textbf{Function~\textsc{UpdateState}($QC, S$)}} 
    \State{~~~~\textbf{Require} ContractState $=$ Active}
    \Statex{~~~~~~~~~~~~~~~$\wedge$ $S.h = height$}
    \Statex{~~~~~~~~~~~~~~~$\wedge~S.preHash = H(st_{height})$}
    \Statex{~~~~~~~~~~~~~~~$\wedge$ \textbf{\textsf{VerifyQC}($QC$)}}
    \State{~~~~$height \leftarrow t+1$}
    \State{~~~~$st_{height}.root \leftarrow S$}
  
\State
\State{\textbf{Function~\textsc{StartChallenge} ($M_{addr}, tx, pledge$)}} 
    \State{~~~~\textbf{Require} ContractState $=$ Active}
    \Statex{~~~~~~~~~~~~~~~$\wedge$ $pledge \geq Pledge_C$}    
    \State{~~~~$id \leftarrow $ H($M_{addr}, tx, block.timestamp$)}
    \State{~~~~$Chal_{id}.tx \leftarrow tx$}
    \State{~~~~$Chal_{id}.startChal \leftarrow block.timestamp$}
    \State{~~~~Trigger Challenge \textbf{event}}  
\State
\State{\textbf{Function~\textsc{ResloveChallenge} ($id, QC$)}} 
    \State{~~~~\textbf{Require} ContractState $=$ Active}
    \Statex{~~~~~~~~~~~~~~~$\wedge$ \textbf{\textsf{VerifyQC}($QC$)}}
    \State{~~~~\textbf{\textsc{Delete}($Chal_{id}$)} }
\State
\State{\textbf{Function~\textsc{SettleRollup} ($id$)}} 
    \State{~~~~\textbf{Require} ContractState $=$ Active}
    \Statex{~~~~~~~~~~~~~~~$\wedge$ $block.time - Chal_{id}.startChal > T_w$}
    \State{~~~~ContractState $\leftarrow$ Frozen}
    \State{~~~~Trigger Settle \textbf{event}}
\State
\State{\textbf{Function~\textsc{SettleWithdraw} ($T_{addr}, M_{addr}, b, \delta, \sigma$)}} 
    \State{~~~~\textbf{Require} ContractState $=$ Frozen}
    \Statex{~~~~~~~~~~~~~~~$\wedge$ \textbf{\textsf{Verify}($T_{sddr}, \sigma$)$ = 1$}}
    \Statex{~~~~~~~~~~~~~~~$\wedge$ \textbf{\textsf{VerifyProof}($\delta, b, T_{sddr}$)$ = 1$}}
    \State{~~~~\textbf{\textsc{Refund}($M_{addr}, b$)} }
\end{algorithmic}
\end{algorithm}

\subsection{Challenge Mechanism}\label{subsec:challenge}
To address the redemption issue, we propose a challenge mechanism deployed in TSC.
In \sysname, clients can initiate a challenge in TSC if their request is not processed.
If there is no response from the sequencers, clients can also withdraw their deposit from TSC with the proof of balance provided by DAP.

First, clients invoke the \textbf{\textsc{StartChallenge}} to initiate a challenge on TSC and receive an $id$ (calculated by the hash of the sender address, transaction, and block timestamp) as a return (cf. Algorithm~\ref{alg:smartcontract} line 21).
Then, TSC sets a timer and triggers a $challenge$ event with $\left<tx, block.timestamp\right>$.
If the enclave receives the $challenge$, it executes the $tx$ in the next batch and votes for it before $\tau$ passes (cf. Algorithm~\ref{alg:smartcontract} line 26-28).
Therefore, when the sequencers provide the QC for the $tx$, it proves that the enclaves have received the transaction and executed it.
{To give the sequencers sufficient time to process the transactions and respond, $\tau$ is set to a duration of several hours.
The waiting time $\tau$ for the challenge does not affect the execution efficiency in normal cases.}

On the other hand, if the sequencer drops the request and cannot generate the response to the challenge in time, clients can invoke the \textbf{\textsc{SettleRollup}}($id$) to settle TSC.
TSC verifies whether the current block time has exceeded the start time of the challenge plus the waiting time $\tau$. 
If the verification is passed, the contract state of TSC becomes frozen (cf. Algorithm~\ref{alg:smartcontract} line 32).
At this time, sequencers can no longer update the state, and the deposit is also not allowed, while the client can withdraw their balance.
To prevent the client's malicious challenge (for example, they do not submit a transaction $tx$ to the sequencer, but carry out a challenge on the chain), the client needs to lock collateral on TSC when initiating a challenge (cf. Algorithm~\ref{alg:smartcontract} line 20).
If the challenge fails, the collateral will be forfeited. 
Additionally, to prevent malicious sequencers from censoring the client's transaction by dropping messages sent to the enclave, the client must submit the $tx$ to the TSC. 
This ensures that the transaction is publicly available and can be processed by the sequencers. 
If the sequencers fail to process the transaction and instead attempt to suppress it by withholding or dropping the transaction, the sequencers will be penalized, and rollup will also be settled.
{If a client finds that its transactions are consistently delayed and require challenges, it can redeem its deposit from \sysname.}

To withdraw the balance $b$ in account $a$, the client needs to provide two proofs: (i) the client has the key of account a; and (ii) the balance in account a is b.
First, according to~\secref{subsec:normal}, the account address $a$ is the public key and the client can generate the signature with the private key, proving their control of account $a$.
Second, the client can get the proof of balance from DAP, \ie, the Merkle proof for the balance of the client's account.
The state recorded in TSC contains the $R$, which is the Merkle root of the account balance.
Thus, the DAPs can provide the balance $b$ and proof $\delta \leftarrow st_{h}.R.proof(a, b)$, where $a$ is the address for the account $b$ in~\sysname.
Client generates the signature $\sigma$ for $\delta$ with the secret key of account $b$ and invoke \textbf{\textsc{SettleWithdraw}($T_{addr}, M_{addr}, b, \delta, \sigma$)}.
Finally, TSC verifies the proof of balance and the signature $\delta$.
If successful, TSC refunds the balance of the account $b$ to the client.
Thus, every client can redeem their deposit on TSC.

The above design solves the problem that the client's transfer and withdrawal (\ie, a transfer to a specific account) transactions are not executed, but the client's deposit is also at risk of not being processed.
Thus, in our design, every $deposit$ should be confirmed by the sequencers.
When clients raise a TTokens request and lock the $deposit$ on the TSC, the $deposit$ is unsolved, and a timer with a countdown of $\tau$ is set.
Then, sequencers process the $deposit$, issue TTokens for the client, and submit the new state to TSC.
Upon the new state being confirmed on the main chain, sequencers respond to the $deposit$, and set it to be solved.å
If the $deposit$ is unsolved until $\tau$ passes, TSC can refund $deposit$ to the clients.

\subsection{Data Availability Provider}\label{subsec:dap}
{
To reduce on-chain storage costs, we adopt an off-chain method for ensuring data availability.
For a state $st_{h}$, DAPs store the transaction lists $txs_h$ and account tree $A_h$, which together constitute the metadata of \sysname.
Alongside the submission of updated states to the TSC, sequencers simultaneously upload the metadata to the DAPs. 
DAPs store the metadata and provide it to clients when settlement occurs (as mentioned in~\secref{subsec:challenge}). 
However, DAPs can act lazily (\eg, withholding data or lazy validation) and pretend as if metadata were stored. 
Therefore, to motivate DAPs to store metadata diligently, we have designed a laziness penalty mechanism to punish the lazy DAP.

}

\bheading{Registration.}
To ensure that DAPs store the metadata of the state recorded on TSC, we require DAPs to register and lock a specified collateral amount.
For clarity, we provide an illustrative example of a DAP registering on the TSC.
Initially, the DAP registers by submitting $\left<D_{addr}, v\right>_{\sigma}$ to TSC, where $D_{addr}$ represents the DAP's address (\ie, the public key), and $v$ is the collateral value.
Upon receipt, TSC authenticates the signature $\sigma$ and validates whether $v$ exceeds the minimum requirement value for collateral, denoted as $C$.
If a DAP successfully passes the verification, it is authorized to store the metadata of \sysname.

{
\bheading{Laziness penalty mechanism.}
We introduce the collateral $C$ and the laziness penalty mechanism to incentivize DAPs.

Specifically, the client can initiate random data requests to the DAPs, and the node that fails to post a response loses part or all of its collateral (\ie, a slashing of collateral).

Once the data request is triggered, the DAPs should respond to the TSC within time $l$, otherwise, the DAPs will suffer a slashing of collateral.
The response consists of the metadata and the Merkle path proof for the metadata.

For the slashing mechanism, we follow the design in~\cite{tas2023cryptoeconomic}, which introduced an optimal slashing function for DAPs.

In our design, for a system with $m$ DAPs, each DAP  stores a copy of the metadata.
Thus, any DAP can respond to the request and provide the complete metadata.
Besides, the response from DAP can be verified by the TSC, and the invalid response is considered to be a no-response.
For a specific data request, the sets $X= \left <x_1,x_2,..., x_m \right>$ represent the reply of DAP.
Let $x_i = 1$, if DAP $q_i$ sends a valid data response to TSC, and $x_i = 0$ otherwise. 
Since this design is a special case for~\cite{tas2023cryptoeconomic}, the slashing function for $q_i$ is defined as follows.
$$f_i(x) =
\begin{cases}
 0,    & i=1\\
 -C, &  \sum_{j=1}^m x_j < 1 ~and~ x_i = 0\\
 -W-\epsilon, & \sum_{j=1}^m x_j \ge 1 ~and~ x_i = 0\\
\end{cases}$$
The TSC slashes the collateral $C$ put up by each DAP that has not sent a valid response before the timeout if there is no valid response. 
Otherwise, if there is more than one response in the TSC before the timeout, it punishes the non-responsive DAPs by a modest amount, namely $W + \epsilon$, where $W$ is the cost for the DAP to construct and send a response to the TSC.
Note that $\epsilon$ is a small number to prevent the DAP from responding negatively to save the cost $W$.
}
\section{Security Analysis} \label{sec:analysis}
We analyze the security of \sysname, which satisfies the \textit{correctness} and \textit{Redeemability} properties.

\begin{lemma}\label{lemmap1g}
If a malicious sequencer forges an invalid current state $st_{h}^{\prime} \neq st_h$, the new state $st_{h+1}^{\prime}$ generated by the enclave will not be accepted by TSC.
\end{lemma}
\begin{proof}
According to the~\secref{subsec:normal}, the honest enclave generates the $st_{h+1}^{\prime}$ containing the hash of $st_{h}^{\prime}$, \ie, $st_{h+1}^{\prime}.preHash \leftarrow H(st_{h}^{\prime})$.
Thus, if the sequencer submits the $st_{h+1}^{\prime}$ signed by the enclave, TSC will check whether $st_{h+1}^{\prime}.preHash$ is equal to $H(st_{h})$.
Obviously, since  $st_{h}^{\prime} \neq st_h$, TSC will not accept the $st_{h+1}^{\prime}$.
\end{proof}

\begin{lemma}\label{lemmap2g}
If a malicious sequencer with a compromised TEE forges an invalid state $st_{h+1}^{\prime}$ with transaction list $txs_{h+1}^{\prime}$
(\ie, $st_{h+1}^{\prime}$ cannot be calculated by executing transactions $txs_{h+1}^{\prime}$ from the initial state $st_{h}$), TSC will not accept the state.
\end{lemma}

\begin{proof}
By Lemma~\ref{lemmap1g}, the sequencer with the protected TEE cannot make the enclave output a forged state by providing the forged input. 
If the sequencer provides the correct input $st_{h}$, the enclave will only generate $st_{h+1}$ by executing $txs_{h+1}^{\prime}$.
This is because the program in the enclave is protected from execution.
Thus, the malicious sequencer with protected TEE cannot generate the forged state $st_{h+1}^{\prime}$.

According to~\secref{sec:model}, \sysname has $n$ sequencers, and up to $f$ TEE can be compromised.
Thus, the $st_{h+1}^{\prime}$ can only be signed by at most $f$ compromised enclaves of sequencers, \ie, it will not be accepted by TSC.
\end{proof}

\begin{theorem}[{\textit{Correctness}}]Malicious sequencers cannot upload an incorrect state accepted by the TSC.

\end{theorem}

\begin{proof}

{
We consider two cases for a malicious sequencer:

\bheading{Case 1: Malicious sequencer with uncompromised TEE}.  
If the sequencer sends the correct current state \( st_h \) to the enclave, the enclave will execute the transactions correctly and generate the correct next state \( st_{h+1} \). 
Thus, the malicious sequencers can only generate the incorrect next state \( st_{h+1}' \) by forging the current state \( st_h' \).
By Lemma~\ref{lemmap1g}, the incorrect next state \( st_{h+1}' \) will be rejected by TSC.

\bheading{Case 2: Malicious sequencer with compromised TEE}.  
By Lemma~\ref{lemmap2g}, if a sequencer forges \( st_{h+1}' \) with its enclave's signature, TSC will also reject the forged state.

In both cases, it is impossible for a malicious sequencer to forge an incorrect state accepted by the TSC.
}
\end{proof}

\begin{theorem}[{\textit{Redeemability}}] If a client submits a transaction $tx$ to sequencers, the transaction will finally be executed or the TSC will be settled (\ie, the deposit of all clients can be redeemed in the main chain). 
\end{theorem}

\begin{proof}
According to the challenge mechanism, the client can initiate a challenge if its transaction is not executed for a long time. 
Once a challenge is initiated, there are two cases as follows.
\begin{packeditemize}
    \item \textbf{Case 1.} Sequencers resolve the challenge before the waiting time $\tau$ passes.
    In this case, at least $f+1$ enclaves of sequencers have received the transaction $tx$.
    Thus, at least one protected enclave will batch the transaction and execute it to generate a new state.
    
    \item \textbf{Case 2.} Sequencers fail to resolve the challenge before the waiting time $\tau$ passes.
    As a result, the client can invoke the \textsc{Settle}($id$) to settle the TSC.
    TSC is frozen immediately, and any client can redeem their deposit on TSC.
\end{packeditemize}

\end{proof}

\section{Performance Evaluation}\label{sec:evaluate}
We evaluate \sysname in terms of on-chain verification costs, throughput, and off-chain transaction processing latency.
We consider two state-of-the-art counterparts, StarkNet~\cite{starknet}, Scroll~\cite{scroll}, Optimism~\cite{optimism}, and Arbitrum~\cite{arbitrum}.

\subsection{System Implementation and Setup}\label{sec:imple}
We build a prototype of \sysname with Golang and develop the smart contracts using Solidity 0.8.0~\cite{solidity}. 
We deploy the smart contracts on the Ethereum test network, Sepolia\footnote{https://github.com/eth-clients/sepolia/}.
Our experiments use mainstream off-the-shelf TEE devices, including Intel SGX~\cite{mckeen2013innovative}, Intel TDX~\cite{tdx}, and Hygon CSV~\cite{csv}. 
We deployed our evaluation on up to 20 nodes, with the distribution of SGX, TDX, and CSV distributed in a ratio of $1:2:2$.
Specifically, nodes with SGX, TDX, and CSV run on Aliyun ECS g7t.2xlarge with 8vCPU (Intel®Xeon), 8i.2xlarge with 8vCPU (Intel®Xeon), and g7h.2xlarge with 8vCPU (Hygon C86-3G 7390), respectively.
We consider two deployment scenarios: local area network (LAN) and wide area network (WAN). 
For WAN experiments, the inter-node Round-Trip Time (RTT) is about 25 ± 0.1ms. 
For LAN experiments, the inter-node RTT is about 0.5 ± 0.03ms.

\begin{table}[t]
    \centering
    \setlength{\abovecaptionskip}{0cm}
    \caption{\textbf{System Parameters.}} \label{table:parameter} 
    \resizebox{\linewidth}{!}{
    \begin{tabular}{@{}lll@{}}
        \toprule[1pt]
        ~~~~~\textbf{Parameters}    & ~~~~~~~\textbf{Values}  \\
        \midrule
        Number of sequencers     & 5, {10}, 15, 20 \\
        Batch size for processing transactions   & 500, 1000, 1500, {2000}     \\
         Number of  SGX sequencers     & 1, {2}, 3, 4 \\
        Number of  TDX sequencers     & 2, {4}, 6, 8 \\
        Number of  CSV  sequencers   & 2, {4}, 6, 8 \\
        \bottomrule[1pt]
    \end{tabular}
    }
\end{table}

We vary parameters in Table~\ref{table:parameter} to evaluate system gas cost, throughput, and latency. 
The experiment evaluates on-chain costs, as well as the throughput and latency of off-chain transaction processing.
Specifically, gas cost measures the expense of executing TSC functions, throughput quantifies the number of transactions handled by sequencers per second, and latency reflects the time required for sequencers to vote and process transactions.

\subsection{On-Chain Cost of Functions} \label{subsec:cost}
The on-chain execution costs of \sysname are measured in \textit{gas}, which is a unit that measures the computational effort required to execute operations on Ethereum~\cite{ethereum}.
The gas consumption depends on the transaction size and execution costs of smart contracts. 
To make these costs more understandable, we convert gas into USD with a gas price\footnote{https://etherscan.io/chart/gasprice} of 19.26 Gwei ($1.926\times10^{-8}$ ETH) and an ETH price\footnote{https://coinmarketcap.com/currencies/ethereum/} of 3376.77 USD on January 1, 2025.

\begin{table}[t]
\centering
\small
\begin{threeparttable}
\caption{\textbf{Protocol Execution Cost.} 
Execution costs are estimated based on the prices (ETH/USD 3376.77 and gas price is 19.26 Gwei) on Jan. 1, 2025}.
\label{table:evalution}
\begin{tabular}{p{5cm}rr}
\toprule

\multirow{2}*{\textbf{~~~~~Methods}}    & \multicolumn{2}{c}{\textbf{Cost}}\\
                          &\textbf{Gas}~~~~~~~&\textbf{USD}\\
\midrule
        \textsc{Deposit}~~~~~~           & 48551~~~~         & 3.16             \\
        \textsc{UpdateState}~~~~~~       & 156,263~~~~        & 10.16     \\
        \textsc{StartChallenge}~~~~~~    & 47,118~~~~        & 3.06       \\
        \textsc{ResolveChallenge}~~~~~~  & 146,618~~~~        & 9.53       \\
        \textsc{SettleRollup}~~~~~~      & 29,078~~~~       & 1.89     \\
        \textsc{SettleWithdraw}~~~~~~    & 124,511~~~~       & 8.10      \\
\midrule
        \textbf{Simple ETH transfer} & 21,000~~~~      & 1.37      \\
\bottomrule
\end{tabular}
\end{threeparttable}
\end{table}

Table~\ref{table:evalution} shows our excellent performance in terms of on-chain cost.
We measure the complete process of deposit, state update, and challenge on Ethereum when the number of sequencers is 10  (see~\secref{sec:design}) and the batch size for processing transactions is 2000.
The \textbf{\textsc{UpdateState}} is one of the most gas-consuming methods, which costs about 156K gas (10.16 USD). 
That is because it includes the verification of the QC and the storage of the new state on-chain.
However, since the cost is amortized over 2000 transactions, the fee per transaction is reduced to just 78 gas (0.005 USD), which is significantly lower than the gas consumption of a simple ETH transfer (1.37 USD).

In terms of the challenge, a challenge-resolve process consists of the \textbf{\textsc{StartChallenge}} and \textbf{\textsc{ResolveChallenge}}, which totally require 194K gas (12.59 USD).
The \textbf{\textsc{SettleRollup}} and \textbf{\textsc{SettleWithdraw}}, which are crucial for the redemption of the client, totally consume 154K gas (9.99 USD).
However, this only occurs when the challenge fails.
In normal cases, the client's withdrawal only requires a transfer in \sysname, which costs 78 gas (0.005 USD).

\subsection{Performance Evaluation}\label{subsec:performance}
We evaluate the performance of \sysname by measuring its throughput and latency under varying numbers of sequencers and batch sizes. 
We also evaluate the impact of the introduced TEE by comparing throughput and latency with sequencers enabled and disabled with TEE.

\bheading{Throughput and latency.}
\figref{fig:performance} illustrates the throughput and latency of \sysname across varying numbers of sequencers and batch sizes in both LAN and WAN environments.
In the LAN environment, as the number of sequencers increases, throughput experiences a slight decrease, while latency rises. 
This behavior is primarily due to the increased messaging overhead among sequencers as their number grows. 
Throughput improves with larger batch sizes, as they enable \sysname to process more transactions concurrently. 
For instance, with 5 sequencers and a batch size of 2000, throughput can reach about 28 KTPS. 
However, larger batch sizes also result in higher latency, which peaks at 119 ms. 
This increase in latency occurs because sequencers need more time to process the larger volume of transactions, leading to longer delays.

In the WAN environment, throughput and latency exhibit similar trends to those in the LAN environment. 
However, due to the higher communication delays inherent in WAN, the throughput is lower and the latency is higher compared to the LAN setting. 
Furthermore, as batch size increases, the rise in latency in the WAN environment is less pronounced than in the LAN environment. 
This difference occurs because, in the WAN, the primary factor influencing latency shifts from transaction processing delays to communication delays.

\begin{figure}[t]
	\centering
	\begin{subfigure}{0.48\linewidth}
		\centering
		\includegraphics[width=1.1\linewidth]{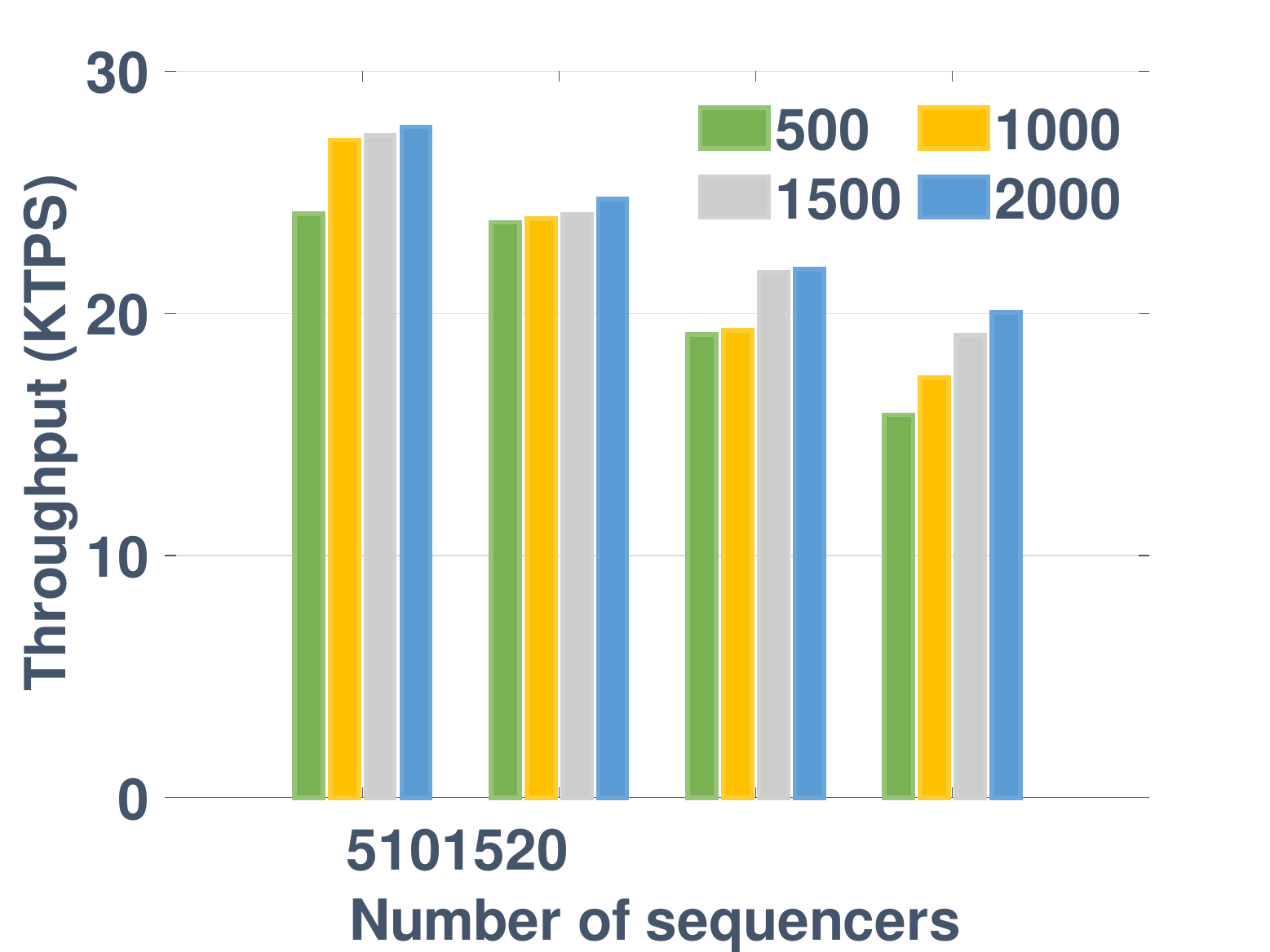}
		\caption{Throughput in LAN}
		\label{fig:throughputLAN}
	\end{subfigure}
	\hfill
	\begin{subfigure}{0.48\linewidth}
		\centering
		\includegraphics[width=1.1\linewidth]{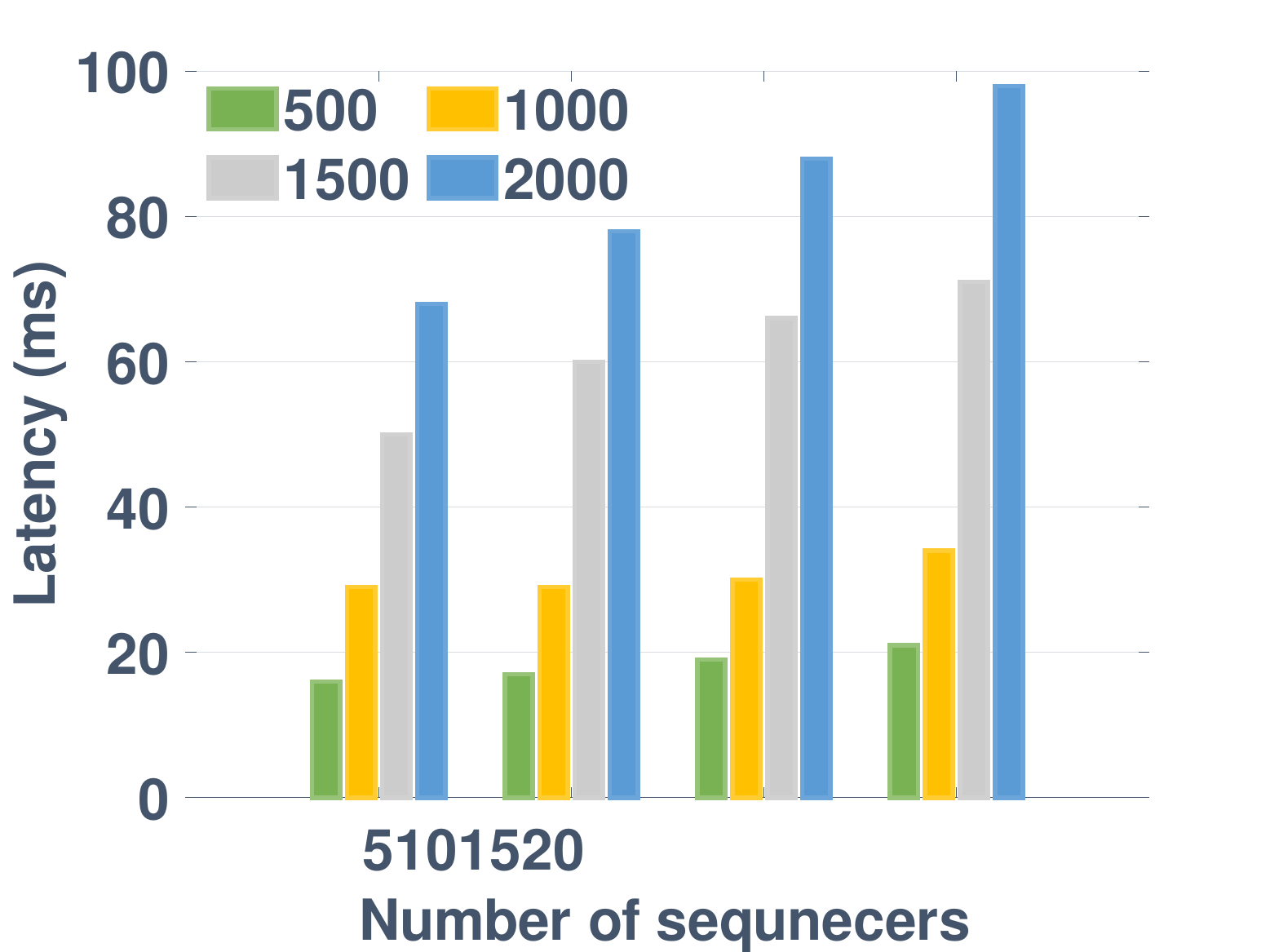}
		\caption{Latency in LAN}
		\label{fig:latencyLAN}
	\end{subfigure}
        \hfill
	\begin{subfigure}{0.48\linewidth}
		\centering
		\includegraphics[width=1.1\linewidth]{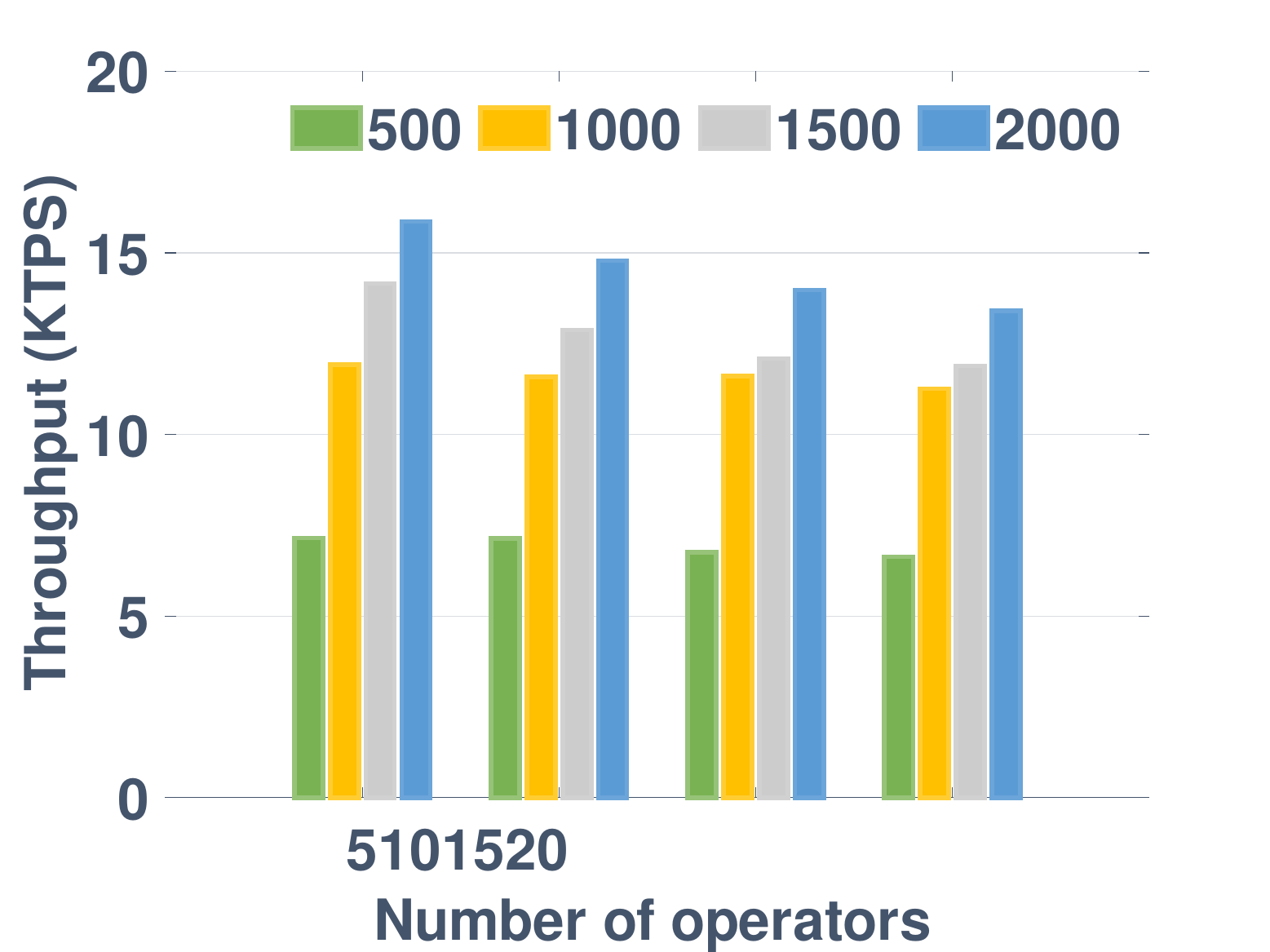}
		\caption{Throughput in WAN}
		\label{fig:throughputWAN}
	\end{subfigure}
	\hfill
	\begin{subfigure}{0.48\linewidth}
		\centering
		\includegraphics[width=1.1\linewidth]{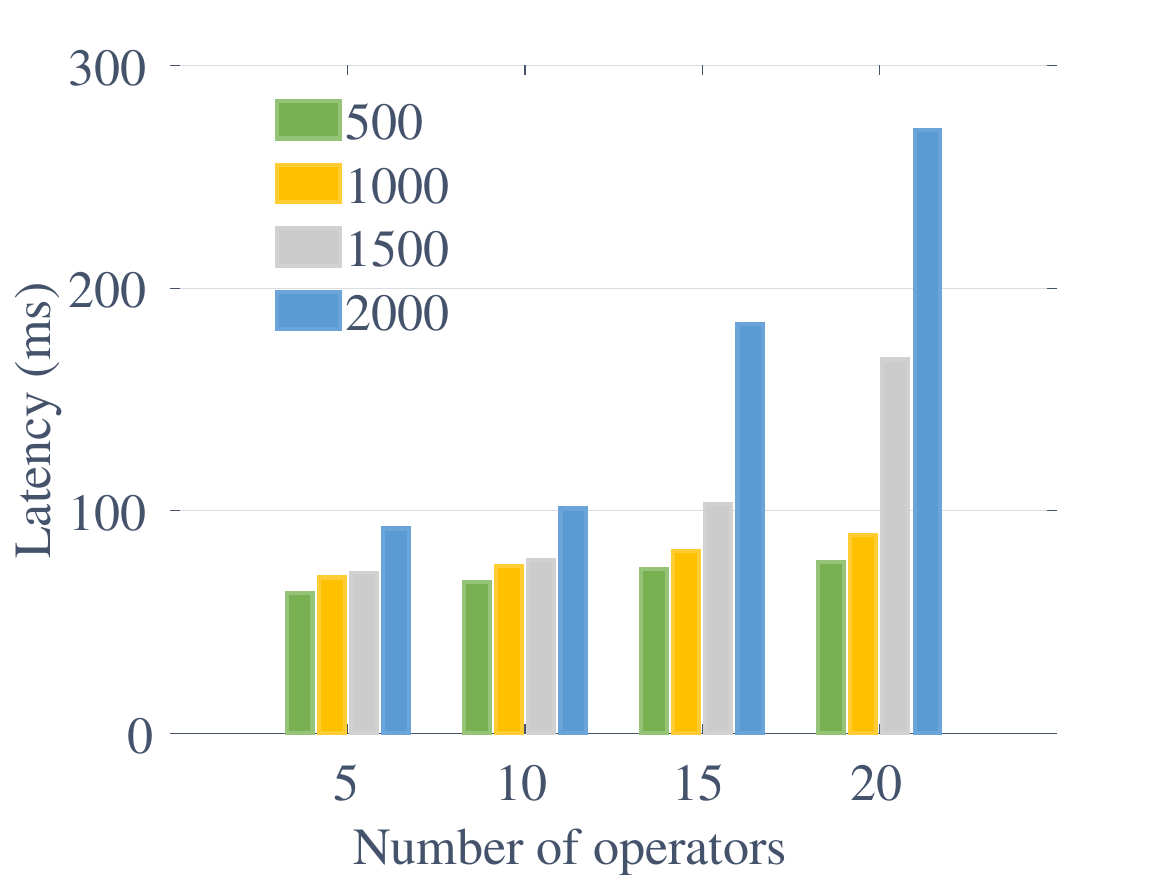}
		\caption{Latency in WAN}
		\label{fig:latencyWAN}
	\end{subfigure}
	\caption{Transaction processing performance with varying numbers of sequencers and batch size in LAN and WAN.}
	\label{fig:performance}
\end{figure}

\bheading{TEE overhead.}
\figref{fig:sev} shows the influence of introducing TEE on throughput and latency in LAN and WAN, respectively.
The operations executed inside a TEE need encryption and to switch context from the regular execution environment of the host machine, which degrades performance.
Thus, the protection provided by TEEs reduces throughput and increases latency.
For example, in the WAN environment with 20 sequencers, throughput decreases by 15.74\%, while latency increases by 20.17\%.
Importantly, even with TEEs' hardware protection in WAN, the reduction in performance does not exceed the acceptable range for users, maintaining a fine balance between security and performance.

\begin{figure}[t]
	\centering
	\begin{subfigure}{0.49\linewidth}
		\centering
		\includegraphics[width=1.1\linewidth]{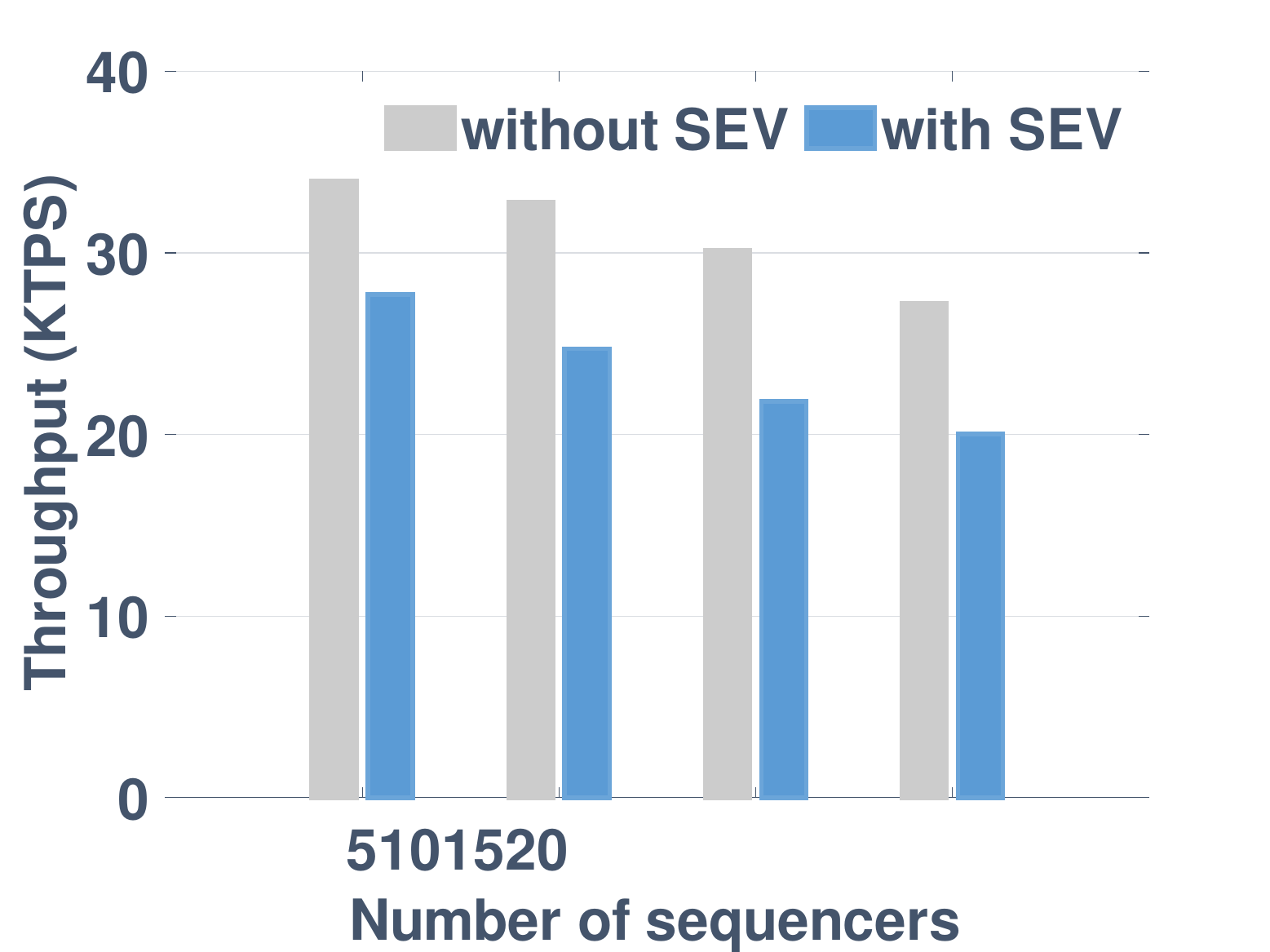}
		\caption{Throughput in LAN}
		\label{fig:throughputnodeLAN}
	\end{subfigure}
	\hfill
	\begin{subfigure}{0.49\linewidth}
		\centering
		\includegraphics[width=1.1\linewidth]{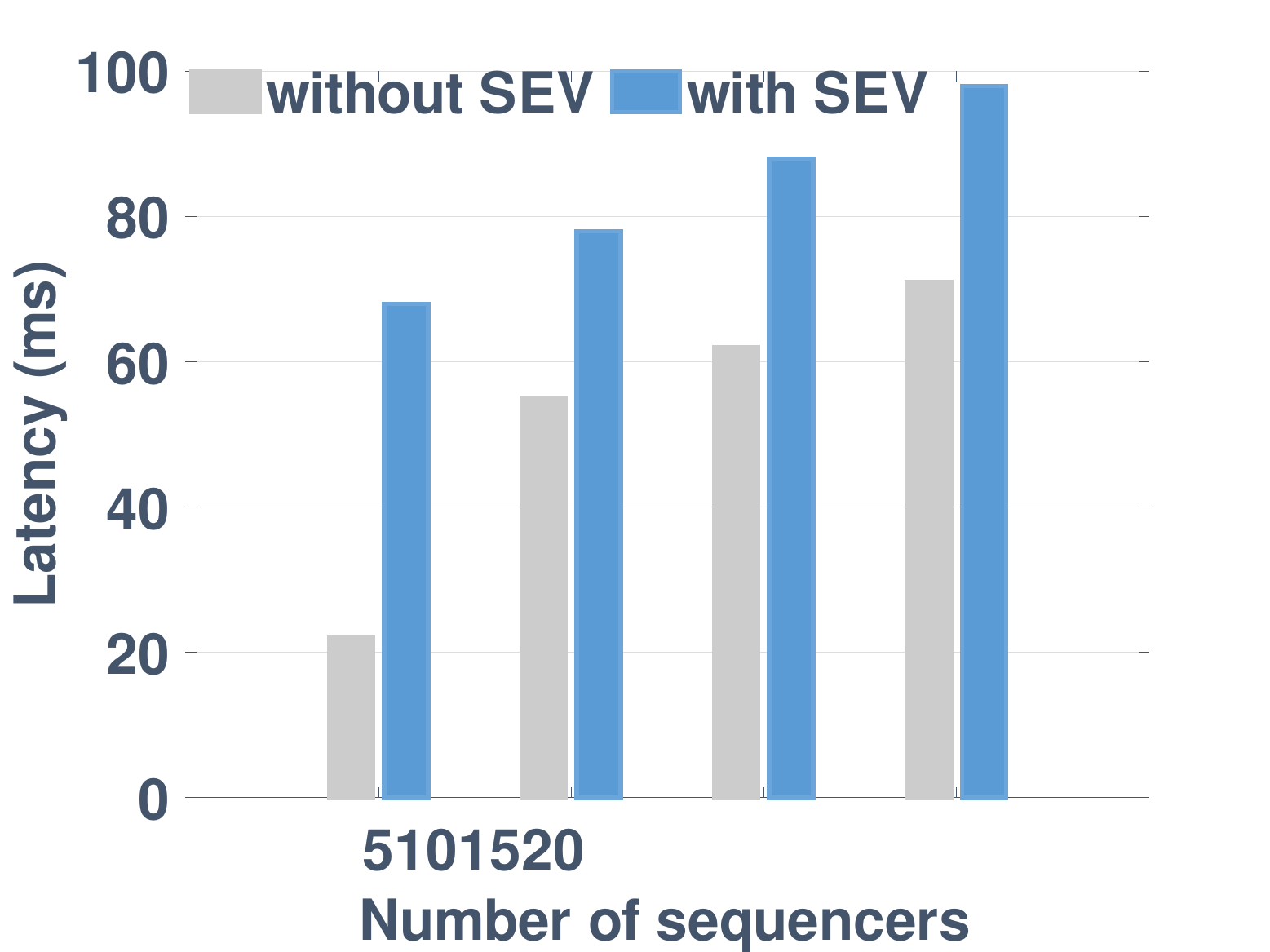}
		\caption{Latency in LAN}
		\label{fig:lantencynodeLAN}
	\end{subfigure}
	\hfill
	\begin{subfigure}{0.49\linewidth}
		\centering
		\includegraphics[width=1.1\linewidth]{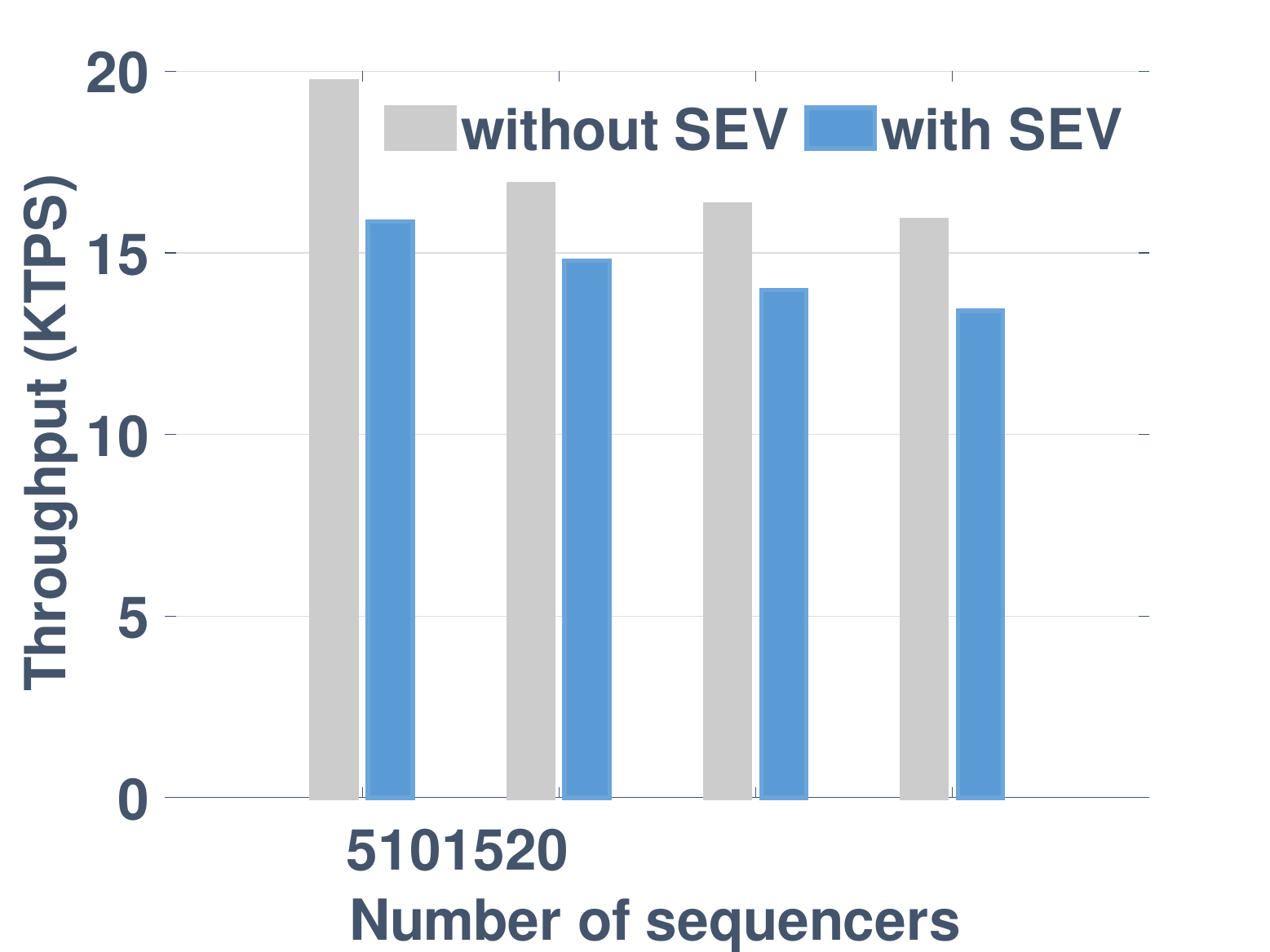}
		\caption{Throughput in WAN}
		\label{fig:throughputnodeWAN}
	\end{subfigure}
	\hfill
	\begin{subfigure}{0.49\linewidth}
		\centering
		\includegraphics[width=1.1\linewidth]{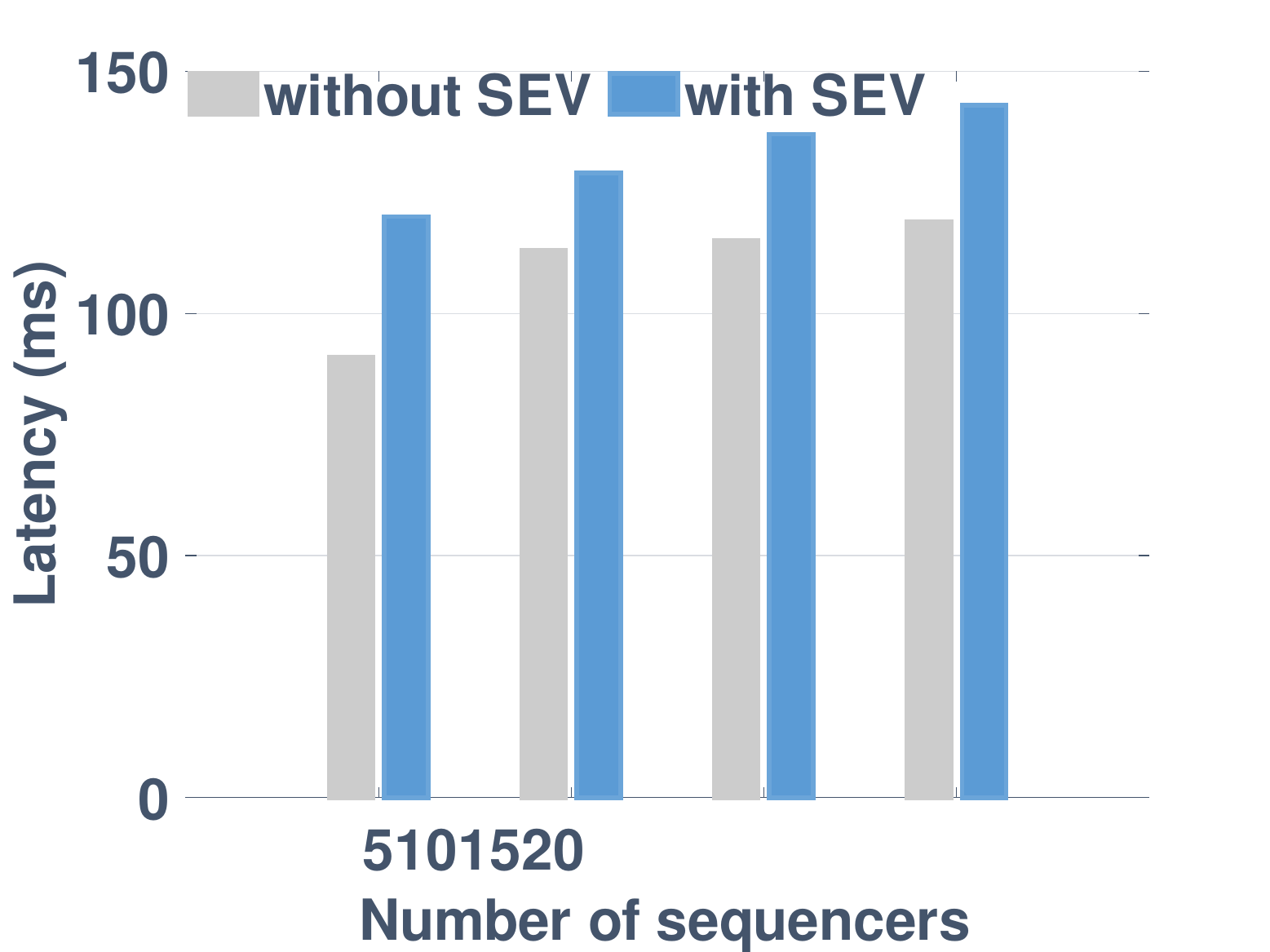}
		\caption{Latency in WAN}
		\label{fig:lantencynodeWAN}
	\end{subfigure}
	\caption{System performance with/without SEV in LAN and WAN.}
	\label{fig:sev}
\end{figure}

\subsection{Comparison with Counterparts}\label{sec:compare}
We conduct a comparison of \sysname against prominent rollup solutions:  StarkNet~\cite{starknet}, Scroll~\cite{scroll}, Optimism~\cite{optimism}, and Arbitrum~\cite{arbitrum}.
The average transaction fees for StarkNet and Scroll over the past month were approximately 0.043 USD and 0.114 USD, respectively~\cite{l2beat}. 
The average transaction fees for Optimism and Arbitrum were significantly lower, at 0.012 USD and 0.005 USD, respectively. 
As shown in~\secref{subsec:cost}, the average transaction fee for \sysname is 0.006 USD per transaction with a batch size of 2000 (matching that of StarkNet), representing an 86\% reduction compared to StarkNet.
Thus, the transaction fee in \sysname is comparable to those of optimistic rollups.
As for throughput, current rollup solutions are constrained by the Ethereum mainnet, which can support up to 3,000 TPS.  
Section~\ref{subsec:performance} demonstrates that \sysname's off-chain processing can achieve throughput exceeding 5000 TPS.
Regarding withdrawal times, StarkNet and Scroll offer withdrawals within a few minutes, while Optimism and Arbitrum have a withdrawal period of one week. 
Although \sysname must undergo a challenge process when TEEs are unavailable, the withdrawal time during the normal case remains a few minutes.

\section{Conclusion}\label{sec:conclusion}
We introduce \sysname, a high-performance and cost-effective rollup solution that leverages TEEs to optimize the gas cost and withdraw delay of rollup. 
\sysname employs a group of sequencers, safeguarded by heterogeneous TEEs, to process transactions off-chain while leveraging DAPs to ensure data availability. 
\sysname also adopts a cost-efficient challenge mechanism to ensure redeemability, even in scenarios where all sequencers become unavailable.
Additionally, \sysname features a laziness penalty mechanism to encourage DAPs to work diligently. 
The prototype implementation of \sysname demonstrates promising results, outperforming state-of-the-art approaches in both on-chain cost efficiency and off-chain processing speed.

\bibliographystyle{IEEEtran}

\bibliography{reference_long}

\end{document}